\documentclass{article}

\usepackage{fullpage,amsmath,amsthm,amssymb}
\usepackage[usenames]{color}

\DeclareMathOperator*{\E}{\mathbb{E}}
\let\Pr\relax
\DeclareMathOperator*{\Pr}{\mathbb{P}}

\usepackage{todonotes}
\newcommand\Oh{\mathcal{O}}

\newcommand{\eps}{\varepsilon}
\renewcommand{\log}{\lg}
\newcommand\bR{\mathbb{R}}

\renewcommand\P{\Pr}
\newcommand{\inprod}[1]{\langle #1 \rangle}

\newcommand\pStable{\mathcal{D}_p}
\newcommand{\eqdef}{\mathbin{\stackrel{\rm def}{=}}}

\newtheorem{theorem}{Theorem}
\newtheorem{lemma}[theorem]{Lemma}

\newtheorem{corollary}[theorem]{Corollary}
\newtheorem{definition}[theorem]{Definition}

\newcommand{\EquationName}[1]{\label{eq:#1}}

\newcommand{\LemmaName}[1]{\label{lem:#1}}

\newcommand{\CorollaryName}[1]{\label{cor:#1}}
\newcommand{\SectionName}[1]{\label{sec:#1}}
\newcommand{\TheoremName}[1]{\label{thm:#1}}

\newcommand{\Equation}[1]{Eq.\:\eqref{eq:#1}}

\newcommand{\Lemma}[1]{Lemma~\ref{lem:#1}}

\newcommand{\Corollary}[1]{Corollary~\ref{cor:#1}}
\newcommand{\Section}[1]{Section~\ref{sec:#1}}
\newcommand{\Theorem}[1]{Theorem~\ref{thm:#1}}

\title{Continuous monitoring of $\ell_p$ norms in data streams}

\author{
  Jaros{\l}aw B{\l}asiok\thanks{Harvard University. \texttt{jblasiok@g.harvard.edu}. Supported by ONR grant N00014-15-1-2388.}
  \and Jian Ding\thanks{University of Chicago. \texttt{jianding@galton.uchicago.edu}. Partially supported by NSF grant DMS-1455049 and an Alfred P.\ Sloan Research Fellowship.}
  \and Jelani Nelson\thanks{Harvard University. \texttt{minilek@seas.harvard.edu}. Supported by NSF grant IIS-1447471 and CAREER award CCF-1350670, ONR Young Investigator award N00014-15-1-2388, and a Google Faculty Research Award.}
}

\begin{document} 

\maketitle

\begin{abstract}
In insertion-only streaming, one sees a sequence of indices $a_1, a_2, \ldots, a_m\in [n]$. The stream defines a sequence of $m$ frequency vectors $x^{(1)},\ldots,x^{(m)}\in\bR^n$ with $(x^{(t)})_i \eqdef |\{j : j\in[t], a_j = i\}|$. That is, $x^{(t)}$ is the frequency vector after seeing the first $t$ items in the stream. Much work in the streaming literature focuses on estimating some function $f(x^{(m)})$. Many applications though require obtaining estimates at time $t$ of $f(x^{(t)})$, for every $t\in[m]$. Naively this guarantee is obtained by devising an algorithm with failure probability $\ll 1/m$, then performing a union bound over all stream updates to guarantee that all $m$ estimates are simultaneously accurate with good probability. When $f(x)$ is some $\ell_p$ norm of $x$, recent works have shown that this union bound is wasteful and better space complexity is possible for the continuous monitoring problem, with the strongest known results being for $p=2$ \cite{HuangTY14,BravermanCIW16,BravermanCINWW17}. In this work, we improve the state of the art for all $0<p<2$, which we obtain via a novel analysis of Indyk's $p$-stable sketch \cite{Indyk06}.
\end{abstract}

\section{Introduction}

Estimating statistics of frequency vectors implicitly defined by insertion-only update streams, as defined in the abstract, was first studied by Flajolet and Martin in \cite{FlajoletM85}. They studied the so-called {\em distinct elements problem}, in which $f(x)$ is the support size of $x$. In the insertion-only model, the support size of $x$ is equivalent to the number of distinct $a_i$ appearing in the stream. One goal in such streaming algorithms, both for this particular distinct elements problem as well as for many others function estimation problems studied in subsequent works, is to minimize the space consumption of the stream-processing algorithm, ideally using $o(n)$ words of memory (note there is always a trivial $n$ space algorithm by storing $x$ explicitly in memory).

For over two decades, work on estimating statistics of frequency vectors of streams remained dormant, until the work of \cite{AlonMS99} on estimating the $p$-norm $\|x\|_p = (\sum_i x_i^p)^{1/p}$ in streams for integer $p\ge 1$. Since then several works have studied these and several other problems, from the perspective of both upper and lower bounds, including estimating $\|x\|_p$ for all $0<p\le 2$ (not necessarily integral) \cite{AlonMS99,Indyk06,IndykW03,Woodruff04,Li08,Li09,KaneNW10b,NelsonW10,KaneNPW11,JayramW13}, $\|x\|_p$ for $p>2$ \cite{AlonMS99,BarYossefJKS04,ChakrabartiKS03,IndykW05,BhuvanagiriGKS06,Gronemeier09,Jayram09,AndoniKO11,BravermanO13,BravermanKSV14,Ganguly15}, empirical entropy \cite{ChakrabartiBM06,ChakrabartiCM10,BhuvanagiriG06,HarveyNO08} and other information-theoretic quantities \cite{IndykM08,GuhaIM08,BravermanO10}, cascaded norms \cite{CormodeM05b,JayramW09,Jayram13}, and several others. There have also been general theorems classifying which statistics of frequency vectors admit space-efficient streaming estimation algorithms \cite{BravermanO10a,BravermanC15,BravermanOR15,BravermanCWY16,BlasiokBCKY15}.

Taking a dynamic data structural viewpoint, ``streaming algorithms'' is simply a synonym for ``dynamic data structures'' but with an implied focus on minimizing memory consumption (typically striving for an algorithm using {\em sublinear} memory). Elements in the stream can be viewed as updates to the frequency vector $x$ (seeing $a\in[n]$ in the stream can be seen as $\texttt{update}$($a,1$), causing the change $x_a\rightarrow x_a + 1$), and the request for an estimate of some statistic of $x$ is a query. In this data structural language, all the works cited in the previous paragraph provide Monte-Carlo guarantees of the following form for queries: starting from any fixed frequency vector and after executing any fixed sequence of updates, the probability that the output of a subsequent query then fails is at most $\delta$. Here we say a query fails if, say, the output is not a good approximation to some particular $f(x)$ (this will be made more formal later). In many applications however, one does not simply want the answer to one query at the end of some large number of updates, but rather one wants to {\em continuously} monitor the data stream. That is, the sequence of data structural operations is an intermingling of updates and queries. For example, one may have a threshold $T$ in mind, and if $f(x)$ ever increases beyond $T$ some data analyst should be alerted. Such a goal could be achieved (approximately) by querying after every update to determine whether the updated frequency vector satisfies this property.  Indeed, the importance of supporting continuous queries in append-only databases (analogous to the insertion-only model of streaming) was recognized 25 years ago in \cite{TerryGNO92}, with several later works focused on continuous stream monitoring with application areas in mind such as trend detection, anomaly detection, financial data analysis, and (bio)sensor data analysis \cite{BabuW01,CarneyCCCLSSTZ02,OlstonJW03}.

If one assumes that a query is being issued after every update, then in a stream of $m$ updates the failure probability should be set to $\delta \ll 1/m$ so that, by a union bound, all queries succeed. Most Monte-Carlo streaming algorithms achieve some space $S$ to achieve failure probability $1/3$, at which point one can achieve failure probability $\delta$ by running $\Theta(\log(1/\delta))$ instantiations of the algorithm in parallel and returning the median estimate (see for example \cite{AlonMS99}). This method increases the space from $S$ to $\Theta(S\log(1/\delta))$, and for many problems (such as $\ell_p$-norm estimation) it is known that at least in the so-called strict turnstile model (i.e.\ $\texttt{update}$($a,\Delta$) is allowed for both positive and negative $\Delta$, but we are promised $x_i\ge 0$ for all $i$ at all times) this form of space blow-up is necessary \cite{JayramW13}. Nevertheless, although improved space lower bounds have been given when desiring that the answer to a {\em single} query fails with probability at most $\delta$, now such blow-up has been shown necessary for the continuous monitoring problem in which one wants, with failure probability $1/3$, to provide simultaneously correct answers for $m$ queries intermingled with $m$ updates. In fact to the contrary, in certain scenarios such as estimating distinct elements or the $\ell_2$-norm in insertion-only streams, improved {\em upper bounds} have been given!

\begin{definition}
We say a Monte-Carlo randomized streaming algorithm $\mathcal{A}$ provides {\em \textbf{strong tracking}} for $f$ in a stream of length $m$ with failure probability $\eta$ if at each time $t\in[m]$, $\mathcal A$ outputs an estimate $\tilde{f}_t$  such that
$$
\Pr(\exists t\in[m] : |\tilde{f}^t - f(x^{(t)})| > \eps f(x^{(t)})) < \eta .
$$
We say that $\mathcal A$ provides {\em \textbf{weak tracking}} for $f$ if 
$$
\Pr(\exists t\in[m] : |\tilde{f}^t - f(x^{(t)})| > \eps \sup_{t'\in[m]} f(x^{(t')})) < \eta .
$$
\end{definition}
Note if $f$ is monotonically increasing, then for insertion-only streams $\sup_{t'\in[m]} f(x^{(t')})$ is simply $f(x^{(m)})$.

The first non-trivial tracking result we are aware of which outperformed the median trick for insertion-only streaming was the \textsc{RoughEstimator} algorithm given in \cite{KaneNW10a} for estimating the number of distinct elements in a stream. \textsc{RoughEstimator} provided a strong tracking guarantee for $f(x) = |\mathop{support}(x)|$ (the distinct elements problem) for constant $\eps, \eta$, using the same space as what is what is required to answer only a single query. This strong tracking algorithm was used as a subroutine in the main {\em non-tracking} algorithm of that work for approximating the number of distinct elements in a data stream up to $1+\eps$.

For $\ell_p$-estimation for $p\in(0,2]$, without tracking, it is known that $\Oh(\eps^{-2}\log(1/\delta))$ words of memory is achievable to return a $(1+\eps)$-approximate value of $f(x) = \|x\|_p$ with failure probability $\delta$ \cite{AlonMS99,Indyk06,KaneNW10b}\footnote{For constant $\delta$ and $p=2$, \cite{AlonMS99} shows that space $\Oh(\eps^{-2}(\log n + \log\log m))$ {\em bits} is achievable in insertion-only streams.}. This upper bound thus implies a strong tracking algorithm with space complexity $\Oh(\eps^{-2}\log m)$ for tracking failure probability $\eta = 1/3$, by setting $\delta < 1/(3m)$ and performing a union bound. The work \cite{HuangTY14} considered the strong tracking variant of $\ell_p$-estimation in insertion-only streams for for any $p$ in the more restricted interval $(1,2]$. They showed that the same algorithms of \cite{AlonMS99,Indyk06}, unchanged, provide strong tracking with $\eta = 1/3$ with space $\Oh(\eps^{-2}(\log n + \log\log m + \log(1/\eps)))$ words\footnote{For $p=2$ their space is as written including the space required to store all hash functions, but for $1 < p < 2$ this space bound assumes that the storage of hash functions is for free.}. This is an improvement over the standard median trick and union bound when the stream length is very long ($m > n^{\omega(1)}$) and $\eps$ is not too small ($\eps > 1/m^{o(1)}$). They also showed that in an update model which allows deletions of items (``turnstile streaming''), any algorithm which only maintains a linear sketch $\Pi x$ of $x$ must use $\Omega(\log m)$ words of memory for constant $\eps$, showing that the median trick is optimal for this restricted class of algorithms.

A different algorithm was given in \cite{BravermanCIW16} for strong tracking for $\ell_2$ using space $\Oh(\eps^{-2}(\log(1/\eps) + \log\log m))$. It was then most recently shown in \cite{BravermanCINWW17} that the AMS sketch itself of \cite{AlonMS99} (though with $8$-wise independent hash functions instead of the original $4$-wise independence proposed in \cite{AlonMS99}) provides strong tracking in space $\Oh(\eps^{-2}\log\log m)$, and weak tracking in space $\Oh(1/\eps^2)$. That is, the AMS sketch provides weak tracking without any asymptotic increase in space complexity over the requirement to correctly answer only a single query.

Despite the progress in upper bounds for tracking $\ell_2$, the only non-trivial improvement for tracking $\ell_p$ is the $\Oh(\eps^{-2}(\log n + \log\log m + \log(1/\eps)))$ upper bound of \cite{HuangTY14}. Although this bound provides an improvement for very long streams ($m$ super-polynomial in $n$), it does not provide any improvement over the standard median trick for the case most commonly studied case in the literature of $m, n$ being polynomially related.

\paragraph{Our contribution.} We show that Indyk's $p$-stable sketch \cite{Indyk06} for $0<p\le 2$, derandomized using bounded independence as in \cite{KaneNW10b}, provides weak tracking while using $\Oh(\log(1/\eps)/\eps^2)$ words of space. It also provides strong tracking using $\Oh(\eps^{-2}(\log\log m + \log(1/\eps))$ words of space. Our bounds thus both improve the space complexity achieved in \cite{HuangTY14} for $\ell_p$-tracking, and well as the range of $p$ supported from $p\in(1,2]$ to all $p\in(0,2]$ (note for $p>2$, it is known that any algorithm requires polynomial space even to obtain a $2$-approximation for a single query, i.e.\ the non-tracking variant of the problem \cite{BarYossefJKS04}).

\section{Notation}
We use $[n]$ for integer $n$ to denote $\{1,\ldots,n\}$. We measure space in words unless stated otherwise, where a single word is at least $\log(nm)$ bits. For $p \in (0,2]$, we let $\pStable$ denote the symmetric $p$-stable distribution, scaled so that for $Z \sim \pStable$, $\P(|Z| > 1) = \frac{1}{2}$. The distribution $\mathcal{D}_p$ has the property that it is supported on the reals, and for any fixed vector $v\in\bR^n$ and $Z_1,\ldots,Z_n,Z$ i.i.d.\ from $\mathcal{D}_p$, $\sum_{i=1}^n Z_i x_i$ is equal in distribution to $\|x\|_p\cdot Z$. See \cite{Nolan17} for further reading on these distributions.

For two vectors $u, v \in \bR^n$ we write $u \preceq v$ to denote coordinatewise comparison, i.e. $u \preceq v$ iff $\forall_i u_i \leq v_i$. For a finite set $S$, we write $\# S$ to denote cardinality of this set.
\section{Preliminaries}

The following lemma is standard. A proof with explicit constants can be found in \cite[Theorem 42]{Nelson11}.

\begin{lemma}
    \LemmaName{p-stable-tails}
	If $Z \sim \pStable$, then $\P(Z > \lambda) \leq \frac{C_p}{\lambda^p}$ for some explicit constant $C_p$ depending only on $p$.
\end{lemma}

We also state some other results we will need.

\begin{lemma}[Paley-Zygmund]
    If $Z \geq 0$ is a random variable with finite variance, then
    \begin{equation*}
        \P(Z > \theta \E Z) \geq (1 - \theta)^2 \frac{(\E Z)^2}{\E(Z^2)}.
    \end{equation*}
\end{lemma}
\begin{corollary}
    \CorollaryName{large-dot-product}
	For fixed vector $v \in \bR^n$, if $\sigma \in \{\pm 1\}^n$ is a vector of 4-wise independent random signs, then
    \begin{equation*}
        \P( \langle \sigma, v \rangle^2 \geq \frac{2}{3}\|v\|_2^2) \geq \frac{1}{27}
        \label{}
    \end{equation*}
\end{corollary}
\begin{proof}
    This follows from $\E \langle \sigma, v\rangle^4 < 3 (\E \langle \sigma, v\rangle^2)^2$ and the Paley-Zygmund inequality.
\end{proof}

\begin{theorem}[{\cite[Theorem 15]{BravermanCIW16, BravermanCINWW17}}]
	Let $v^{(1)}, v^{(2)}, \ldots v^{(m)} \in \bR^n$, be a sequence of vectors such that $0 \preceq v^{(1)} \preceq v^{(2)} \preceq \ldots \preceq v^{(m)}$. Let $\sigma \in \{\pm 1\}^n$ be a vector of 4-wise independent random signs.
    Then 
	\begin{equation*}
		\P\left( \sup_{i \leq m} |\langle \sigma, v^{(i)} \rangle| > \lambda \|v^{(n)}\|_2\right) < \frac{C}{\lambda^2}
	\end{equation*}
	for some universal constant $C$.
    \TheoremName{chaining}
\end{theorem}

\begin{theorem}{{\cite{KaneNW10b,DiakonikolasKN10}}}
	\TheoremName{kwise-pstable}
	If $Z_i \sim \pStable$ for $i \in [n]$ are $k$-wise independent random variables, then for every vector $x \in \bR^n$ and every pair $a, b \in \bR \cup \{\pm\infty\}$ we have
	\begin{equation*}
		\P(\inprod{Z, x} \in (a,b) ) = \P(\|x\|_p Z_1 \in (a,b)) \pm \Oh(k^{-1/p})
		\label{}
	\end{equation*}
\end{theorem}

\begin{theorem}{\cite[Lemma 2.3]{BellareR94}}
	\TheoremName{chernoff}
	Let $X_1, \ldots X_n \in \{0, 1\}$ be a sequence of $k$-wise independent random variables, and let $\mu = \sum \E X_i$. Then
	\begin{equation*}
		\forall\lambda > 0,\ \P(\sum X_i \geq (1 + \lambda) \mu) \leq \exp(-\Omega(\min\{\lambda, \lambda^2\} \mu)) + \exp(-\Omega(k))
		\label{}
	\end{equation*}
\end{theorem}

\section{Overview of approach}
Indyk's $p$-stable sketch picks a random matrix $\Pi \in \bR^{d \times n}$ such that each entry is drawn according to the distribution $\pStable$. It then maintains the sketch $\Pi x^{(t)}$ of the current frequency vector. This sketch can be easily updated as the frequency vector changes, i.e.\ after observing an index $a_j \in [n]$ we update the sketch by $\Pi x^{(t+1)} := \Pi x^{(t)} + \Pi e_{a_j}$. An $\|x\|_p$-estimate query is answered by returning the median of $|\Pi x^{(i)}|_j$ over $j\in[d]$. Since storing $\Pi$ in memory explicitly is prohibitively expensive, we generate it so that the entries in each row are $k$-wise independent for $k = \Oh(1/\eps^p)$ (as done in \cite{KaneNW10b}), and the $d$ seeds used to generate the rows of $\Pi$ are $\Oh(\log(1/(\eps\delta)))$-wise independent. We also work with discretized $p$-stable random variables to take bounded memory. All together, the bounded independence and discretization, also performed in \cite{KaneNW10b}, allow us to store $\Pi$ using low memory.

We then show that instantiating Indyk's algorithm with $d = \Oh(\eps^{-2}\log(1/(\eps\delta)))$ provides the weak tracking guarantee with failure probability $\delta$. The analysis of the correctness of this algorithm is as follows. Let $\pi_i$ denote the $i$th row of $\Pi$. We first show a result resembling the Doob's martingale inequality --- namely, in \Section{analysis} we show that for a fixed $i$, if we look at the evolution of $\inprod{ \pi_i, x^{(t)}}$ as $t$ increases, the largest attained value $(\sup_{t \leq m} \inprod{\pi_i, x^{(t)}})$ is with good probability not much larger than the median of the distribution $|\inprod{\pi_i, x^{(m)}}|$, which is the typical magnitude of the counter at the end of the stream. This fact resembles similar facts shown in \cite{BravermanCIW16,BravermanCINWW17} for when the $\pi_i$ have independent Rademachers as entries, though our situation is complicated by the fact that $p$-stable random variables have much heavier tails.

We then, discussed in \Section{weak-tracking}, show how the previous paragraph implies a weak tracking algorithm with $d = \Oh(\eps^{-2}\log(1/(\eps\delta)))$: we split the sequence of updates into $\mathop{poly}(1/\eps)$ intervals such that the $\ell_p$-norm of the frequency vector of updates in each of those intervals, i.e.\ $\|x^{(t+1)} - x^{(t)}\|_p$, is of the order $\varepsilon^{\Theta(1)} \|x^{(m)}\|_p$. We then union bound over the $\mathop{poly}(1/\eps)$ intervals to argue that the algorithm's estimate is good at each of the interval endpoints. This is the source of the extra factor of $\log(1/\eps)$ in our space bound: to obtain $\eps^{-\Omega(1)}$ failure probability to union bound over these intervals. On the other hand, within each of the intervals most of the counters do not change too rapidly by the argument developed in \Section{analysis}.

Finally, in \Section{strong-tracking} we show how given an algorithm satisfying a weak tracking guarantee, one can use it to get a strong-tracking algorithm with slightly larger space complexity. This argument was already present in \cite{BravermanCINWW17}. One first identifies $q$ points in the input stream at which the $\ell_p$ norm roughly doubles when compared to the previously marked point. There are only $\Oh(\log m)$ such intervals. It is then enough to ensure that our algorithm satisfies weak tracking for all those $\Oh(\log m)$ prefixes simultaneously, in order to deduce that the algorithm in fact satisfies strong tracking. This is done by union bound over $\Oh(\log m)$ bad events (as opposed to standard union bound over $\Oh(m)$ bad events), which introduces an extra $\log\log m$ factor in the space complexity as when compared to weak tracking.

\section{Analysis \SectionName{analysis}}

We first show two lemmas that play a crucial role in our weak tracking analysis.

\begin{lemma}
	Let $x \in \bR^n$ be a fixed vector, and $Z \in \bR^n$ be a random vector with $k$-wise independent entries drawn according to $\pStable$. Then
    \begin{equation*}
		\P(\sum x_i^2 Z_i^2 \geq \lambda^2 \|x\|_p^2) \leq \frac{C}{\lambda^p} + \Oh(k^{-1/p})
        \label{}
    \end{equation*}
    for some universal constant $C$.
    \LemmaName{sos}
\end{lemma}
\begin{proof}
    Let $E_0$ be the event $\sum x_i^2 Z_i^2 \geq \lambda^2 \|x\|_p^2$. Note that $E_0$ depends only on $|Z_i|$, and does not depend on the signs of the $Z_i$. We write $Z_i = |Z_i| \sigma_i$, where $\sigma_i$ are $k$-wise independent random signs. Conditioning on $|Z_i|$,
    \begin{equation*}
    \E_\sigma\left( \left(\sum x_i |Z_i| \sigma_i\right)^2 \middle| |Z_1|, \ldots |Z_n| \right) = \sum x_i^2 Z_i^2
        \label{}
    \end{equation*}
    and therefore for any $|Z_1|,\ldots,|Z_m|$ for which $E_0$ holds, by \Corollary{large-dot-product}
    \begin{equation*}
        \P_{\sigma}\left(  \left( \sum x_i |Z_i| \sigma_i \right)^2 \geq \frac{2}{3}\lambda^2 \|x\|_p^2 \middle| |Z_1|,\ldots,|Z_m| \right) \geq 
        \P_{\sigma}\left(  \left( \sum x_i |Z_i| \sigma_i \right)^2 \geq \frac{2}{3}\sum x_i^2 Z_i^2 \middle| |Z_1|,\ldots,|Z_m|\right) \geq 
        \frac{1}{27}
        \label{}
    \end{equation*}
    and thus
    \begin{equation*}
        \P_\sigma\left(  \left( \sum x_i |Z_i| \sigma_i \right)^2 \geq \frac{2}{3}\lambda^2 \|x\|_p^2 \middle| |Z_1|, \ldots |Z_n| \right) \geq \frac{\mathbf{1}_{E_0}}{27} ,
    \end{equation*}
    where $\mathbf{1}_{E_0}$ is an indicator random variable for event $E_0$. Integrating over $|Z_i|$, 
    \begin{equation}\label{eqn:e0}
        \P_{\sigma,Z}\left( \left( \sum x_i |Z_i| \sigma_i \right)^2 \geq \frac{2}{3} \lambda^2 \|x\|_p^2 \right) \geq \frac{1}{27} \P_Z(E_0) .
    \end{equation}

     On the other hand $|Z_i| \sigma_i$ has the same distribution as $Z_i$, and moreover
    \begin{align}
\nonumber        \P_Z\left( \left(  \sum x_i Z_i \right)^2 \geq \frac{2}{3} \lambda^2 \|x\|_p^2 \right) & = \P_Z\left( \left| \langle x, Z \rangle \right| \geq \sqrt{\frac{2}{3}} \lambda \|v\|_p \right) \\
\nonumber		& \leq \P_Z\left( \|x\|_p \tilde{Z} \geq \sqrt{\frac{2}{3}} \lambda \|x\|_p\right) + \Oh(k^{-1/p}) \\
		& \leq \frac{C}{\lambda^p} + \Oh(k^{-1/p}) \label{eqn:pstable-decay}
    \end{align}
	where $\tilde{Z}\sim\mathcal{D}_p$. The inequalities are obtained via \Theorem{kwise-pstable} and \Lemma{p-stable-tails}. Combining \eqref{eqn:e0}, \eqref{eqn:pstable-decay} yields
    \begin{equation*}
		\P_Z(E_0) \leq \frac{27C}{\lambda^p} + \Oh(k^{-1/p}) .
    \end{equation*}
\end{proof}

\begin{lemma}
	Let $x^{(1)}, x^{(2)}, \ldots x^{(m)} \in \bR^n$ satisfy $0 \preceq x^{(1)} \preceq x^{(2)} \preceq \ldots \preceq x^{(m)}$. Let $Z \in \bR^n$ have $k$-wise independent entries marginally distributed according to $\pStable$. Then for some $C_p$ depending only on $p$,
    \begin{equation*}
		\P\left(\sup_{k \leq m} |\langle Z, x^{(k)}\rangle| \geq \lambda \|x^{(m)}\|_p\right) \leq C_p \left(\frac{1}{\lambda^{2p/(2+p)}} + k^{-1/p}\right) .
        \label{}
    \end{equation*}
	\LemmaName{tracking}
\end{lemma}
\begin{proof}
    Observe that for any $\beta$ we have 
    \begin{align*}
		\P\left(\sup_{k \leq m} |\langle Z, x^{(k)}\rangle| \geq \lambda \|x^{(m)}\|_p\right)  
		& \leq \P\left( \sum_i Z_i^2(x^{(m)})_i^2 \geq \beta^2 \|x^{(m)}\|_p^2\right) \\
        & + \P\left(\sup_{k \leq m} |\langle Z, x^{(k)}\rangle| \geq \lambda \|x^{(m)}\|_p \middle| \sum Z_i^2 (x^{(m)})_i^2 < \beta^2 \|x^{(m)}\|_p^2\right) .
    \end{align*}

    \Lemma{sos} directly implies that 
    \begin{equation}
		\P\left(\sum Z_i^2 (x^{(m)})_i^2 \geq \beta^2 \|x^{(m)}\|_p^2\right) \leq \frac{C}{\beta^p} + \frac{C}{k^{1/p}}.
        \EquationName{first-part}
    \end{equation}

    On the other hand we can write $Z_i = |Z_i| \sigma_i$, where $\sigma_i$ are $k$-wise independent Rademacher random variables, independent from $|Z_i|$. Let us define $w^{(k)} \in \mathbb{R}^n$ for $k\in [m]$ to be the vector with coordinates $(w^{(k)})_i := (x^{(k)})_i |Z_i|$, so that $\inprod{x^{(k)}, Z} = \inprod{w^{(k)}, \sigma}$, and in particular 
	\begin{equation*}
		\sup_{k\leq m} \left|\inprod{Z, x^{(i)}}\right| = \sup_{k \leq m} \left|\inprod{\sigma, w^{(i)}}\right|.
		\label{}
	\end{equation*}
	
	Now, if we condition on $|Z_1|, \ldots |Z_n|$, then the sequence $w^{(1)}, \ldots w^{(k)}$ of vectors satisfies the assumptions of \Theorem{chaining}, and we can conclude that
	\begin{equation*}
        \P\left(\sup_{k \leq m} \left|\inprod{\sigma, w^{(k)}}\right| > \frac{\lambda}{\beta} \|w^{(m)}\|_2\right) \leq \frac{C\beta^2}{\lambda^2} .
	\end{equation*}
    Moreover if $|Z_i|$ are such that $\sum Z_i^2 (x^{(m)})_i^2 \leq \beta^2 \|x^{(m)}\|_p^2$, or equivalently $\|w^{(m)}\|_2^2 \leq \beta^2 \|x^{(m)}\|_p^2$, we have
	\begin{equation*}
		\P\left(\sup_{k \leq m} \left|\inprod{\sigma, w^{(k)}}\right| > \lambda \|x^{(m)}\|_p\right) \leq \frac{C\beta^2}{\lambda^2} ,
	\end{equation*}
	which implies
    \begin{equation*}
        \P\left(\sup_{k \leq m} |\langle Z, x^{(k)}\rangle| \geq \lambda\|x^{(m)}\|_p \,\middle|\, \sum (Z_i x^{(m)}_i)^2 < \beta \|x^{(m)}\|_p^2\right)  \leq \frac{C\beta^2}{\lambda^2}.
    \end{equation*}

    This together with \Equation{first-part} yields 
	\begin{equation*}
		\P\left(\sup_{k \leq m} |\langle Z, x^{(k)}\rangle| \geq \lambda \|x^{(m)}\|_p\right) \leq \frac{1}{\beta^p} + \frac{C \beta^2}{\lambda^2} + \frac{C}{k^{1/p}}
		\label{}
	\end{equation*}
	We can take $\beta := \Theta(\lambda^{\frac{2}{2+p}})$, to have $\frac{1}{\beta^p} + \frac{C\beta^2}{\lambda^2} = \Oh(\lambda^{-\frac{2p}{2+p}})$. 
\end{proof}

\subsection{Weak tracking of $\|x\|_p$}\SectionName{weak-tracking}

In this section we upper bound the number of rows needed in Indyk's $p$-stable sketch with boundedly independent entries to achieve weak tracking.

\begin{lemma}
    \LemmaName{weak-tracking}
    Let $x^{(1)}, \ldots x^{(m)} \in \bR^{n}$ be any sequence satisfying $0 \preceq x^{(1)} \preceq x^{(2)} \preceq \ldots \preceq x^{(m)}$.
	Take $\Pi \in \bR^{d \times n}$ to be a random matrix with entries drawn according to $\pStable$, and such that the rows are $r$-wise independent, and all entries within a row are $s$-wise independent. 
	
	For every $k \in [m]$, define $\alpha_k$ to be $\mathop{median}\left( |(\Pi x^{(k)})_1|, \ldots, |(\Pi x^{(k)})_d|\right)$.
	If $d = \Omega(\varepsilon^{-2} (\log \frac{1}{\varepsilon} + \log \frac{1}{\delta})), r = \Omega(\log\frac{1}{\varepsilon} + \log\frac{1}{\delta})$ and $s = \Omega(\varepsilon^{-p})$, then with probability at least $1 - \delta$ we have
    \begin{equation*}
        \forall k\in[m],\ \|x^{(k)}\|_p - \varepsilon \|x^{(m)}\|_p \leq \alpha_k \leq \|x^{(k)}\|_p + \varepsilon \|x^{(m)}\|_p
        \label{}
    \end{equation*}
\end{lemma}
\begin{proof}
    Consider a sequence of indices $1 < t_1 < t_2 < \ldots  < t_{q+1} = m$, constructed inductively in the following way.  We take $t_1$ to be the smallest index with $\|x^{(t_1)}\|_p \geq \varepsilon^{4/p}\|x^{(m)}\|_p$. Given $t_k$, we take $t_{k+1}$ to be the smallest index such that $\|x^{(t_{k+1})} - x^{(t_{k})}\|_p \geq \varepsilon^{4/p}\|x^{(m)}\|_p$ if there exists one, and $t_{k+1}=m$ otherwise.

	Observe that $q \leq \eps^{-8/p}$. Indeed, for $p \geq 1$ we have
    \begin{equation*}
		\|x^{(m)}\|_p^p = \|x^{(t_1)} + \sum_{1 \leq i < q} (x^{(t_{i+1})} - x^{(t_i)})\|_p^p \geq \|x^{(t_1)}\|_p^p + \sum_{1 \leq i <  q} \|x^{(t_{i+1})} - x^{(t_i)}\|_p^p \geq q \varepsilon^{4} \|x^{(m)}\|_p^p
    \end{equation*}
    where the inequality $\|x^{(t_1)} + \sum_{i\geq 1} (x^{(t_{i+1})} - x^{(t_i)})\|_p^p \geq \|x^{(t_1)}\|_p^p + \sum \|x^{(t_{i+1})} - x^{(t_i)}\|_p^p$ holds because all vectors $x^{(1)}$ and $x^{(t_{i+1})} - x^{(t_i)}$ for every $i$ have non-negative entries --- we can consider each coordinate separately, and use the fact that for $p \geq 1$ and nonnegative numbers $a_i$ we have $(\sum a_i)^p \geq \sum a_i^p$ --- or equivalently, $\|a\|_1^p \geq \|a\|_p^p$. After rearranging this yields $q \leq \varepsilon^{-4} \leq \varepsilon^{-8/p}$.
	
	Similarly, for $p \leq 1$, we have that for non-negative numbers $a_i$, $(\sum_{i \leq q} a_i)^p \geq q^{p-1} \sum a_i^p$ (this is true because for fixed $\sum a_i$, the sum $\sum a_i^p$ is maximized when all $a_i$ are equal), and therefore
	\begin{equation*}
        \|x^{(m)}\|_p^p = \|x^{(t_1)} + \sum_{1 \leq i < q} (x^{(t_{i+1})} - x^{(t_i)})\|_p^p \geq q^{p-1} \left(\|x^{(t_1)}\|_p^p + \sum_{1 \leq i <  q} \|x^{(t_{i+1})} - x^{(t_i)}\|_p^p\right) \geq q^p \varepsilon^{4} \|x^{(m)}\|_p^p
	\end{equation*}
	which implies $q \leq \varepsilon^{-4/p}$.

	For $j \in [m]$, let us define
	\begin{align*}
		l_j & := \#\{ i : |\inprod{\pi_i, x^{(j)}}| < (1-\varepsilon) \|x^{(j)}\|_p\} \\ 		u_j & := \#\{ i : |\inprod{\pi_i, x^{(j)}}| > (1+\varepsilon) \|x^{(j)}\|_p\} .
	\end{align*}
	
	Let $\tilde{\pi}_i$ be a vector of i.i.d. random variables drawn according to $\pStable$. We know that $\inprod{\tilde{\pi}_i, x^{(j)}} \sim \|x^{(j)}\|_p \pStable$. Hence $\P(|\inprod{\tilde{\pi}_i, x^{(j)}}| > \|x^{(j)}\|_p) = \frac{1}{2}$, and $\P(|\inprod{\tilde{\pi}_i, x^{(j)}}| > (1 + \varepsilon) \|x^{(j)}\|_p) \leq \frac{1}{2} - 2 C \varepsilon$ for some universal constant $C$. Similarly $\P(|\inprod{\tilde{\pi}_i, x^{(j)}}| < (1-\varepsilon)\|x^{(j)}\|_p) \leq \frac{1}{2} - 2 C \varepsilon$. 

	Entries of $\pi_i$ are $s$-wise independent, for $s \geq C_2 \varepsilon^{-p}$ with some large constant $C_2$ depending on $C$. Thus by \Theorem{kwise-pstable}, $\P(|\inprod{\pi_i, x^{(j)}}| < (1-\varepsilon)\|x^{(j)}\|_p) \leq \P(|\inprod{\tilde{\pi}_i, x^{(j)}}| < (1-\varepsilon)\|x^{(j)}\|_p) + C \varepsilon \leq \frac{1}{2} - C \varepsilon$, and analogously for $\P(|\inprod{\pi_i, x^{(j)}}| > (1 + \varepsilon)\|x^{(j)}\|_p) < \frac{1}{2} - C \varepsilon$.

	Hence 
	\begin{align*}
		\E l_j &\leq d\left(\frac{1}{2} - C \varepsilon\right)\\
		\E u_j &\leq d\left(\frac{1}{2} - C \varepsilon\right).
	\end{align*}
	
	For $j \in [q]$, let $S_j$ be the event
	\begin{equation*}
		\left\{ l_{t_j} \leq \frac{d}{2} - \frac{C d}{2} \varepsilon \right\} \land \left\{u_{t_j} \leq \frac{d}{2} - \frac{C d}{2} \varepsilon\right\}
	\end{equation*}
	
	Note that for fixed $j$ and varying $i$, indicator random variables for the events ``$|\inprod{\pi_i, x^{(j)}}| < (1-\varepsilon) \|x^{(j)}\|_p$''  are $r$-wise independent. Thus by \Theorem{chernoff}, $\P(S_j) \geq 1 - C' \exp( -\Omega(d\varepsilon^2)) - \exp(-\Omega(r))$. Taking $d = \Omega_p(\varepsilon^{-2} (\log{\frac{1}{\varepsilon}} + \log \frac{1}{\delta}))$ and $r = \Omega_p(\log \frac{1}{\varepsilon\delta})$ we can get $\P(S_j) \geq 1 - \frac{\delta \varepsilon^{8/p}}{2}$, and hence by a union bound all $S_j$ hold simultaneously except with probability at most $\frac{\delta}{2}$ since the number of events $S_j$ is $q \leq \varepsilon^{-8/p}$.

	For $i \in [d]$ and $j \in [q]$, let $E_{i,j}$ be the event
	\begin{equation*}
		\exists t \in [t_j, t_{j+1} - 1],\ |\inprod{x^{(t)} - x^{(t_j)}, \pi_i}| > \varepsilon \|x^{(m)}\|_p.
	\end{equation*}

    By construction of the sequence $t_j$, all $x^{(t)} - x^{(t_j)}$ above have $\ell_p$ norm at most $\varepsilon^4\|x^{(m)}\|_p$, we can invoke \Lemma{tracking} to deduce that $\P(E_{ij}) \leq C_3 \left(\frac{\varepsilon^{4/p}}{\varepsilon}\right)^{\frac{2p}{2+p}} + C_3 s^{-1/p} \leq C_3 \varepsilon + C_3 s^{-1/p}$. Again if we pick $s \geq C_4 \varepsilon^{-p}$ for sufficiently large $C_4$ and small enough $\varepsilon$ we have $\P(E_{ij}) \leq \frac{C}{4} \varepsilon$. Therefore for any fixed $j$, we have
	\begin{equation*}
		\E \sum_{i} \mathbf{1}_{E_{ij}} \leq \frac{C}{4} d\varepsilon
	\end{equation*}
	And finally again by \Theorem{chernoff}, for each $j$
	\begin{equation*}
		\P(\sum_i \mathbf{1}_{E_{ij}} \geq \frac{C}{2} d\varepsilon) \lesssim \exp(-C' d \varepsilon) + \exp(-C' r) 
	\end{equation*}

	We have $d \geq C_3 \varepsilon^{-2} \log \frac{1}{\delta\varepsilon}$, and $q \leq \varepsilon^{-8/p}$, hence for sufficiently small $\varepsilon$, we have $\exp(-C' d \varepsilon) \leq \frac{\delta}{2q}$. On the other hand if $r = \Omega_p(\log \frac{1}{\delta\varepsilon})$ is sufficiently large, we have $\exp(-C'r) \leq \frac{\delta}{2q}$. We invoke the union bound over all $j$ to deduce that with probability at least $1 - \frac{\delta}{2}$ the following event $V$ holds:
	\begin{equation*}
		\forall j,\ \sum_i \mathbf{1}_{E_{ij}} \leq \frac{C}{2} d \varepsilon .
	\end{equation*}

	We know that with probability at least $1 - \delta$ simultaneously $V$ and all the events $S_j$ hold. We will show now that, when these events all hold, then $\forall k\ \|x^{(k)}\|_p - K \varepsilon \|x^{(m)}\|_p \leq \alpha_k \leq \|x^{(k)}\|_p + K \varepsilon \|x^{(m)}\|_p$ for some universal constant $K$. Indeed, consider some $k$, and let us assume that $t_j \leq k \leq t_{j+1}$. With event $S_j$ satisfied, we know that $\#\{ i : |\inprod{\pi_i, x^{(t_{j})}}| \leq \|x^{(t_{j)}}\|_p + \varepsilon \|x^{(m)}\|_p\} \geq d\left( \frac{1}{2} + \frac{C\varepsilon}{2} \right)$, and with event $V$ satisfied, we know that for all but $\frac{C\varepsilon}{2}d$ of indices $i$ we have $|\inprod{\pi_i, x^{(k)} - x^{(t_{j})}}| \leq \varepsilon \|x^{(m)}\|$.

    By the triangle inequality $|\inprod{\pi_i, x^{(k)}}| \leq |\inprod{\pi_i, x^{(t_j)}}| + |\inprod{\pi_i, x^{(k)} - x^{(t_j)}}|$, yielding
	\begin{equation*}
		\#\{ i : |\inprod{\pi_i, v_k}| \leq \|v_{t_{j}}\|_p + 2 \varepsilon \|v_m\|_p\} \geq \frac{d}{2}.
		\label{}
	\end{equation*}
	With similar reasoning we can deduce that
	\begin{equation*}
		\#\{ i : |\inprod{\pi_i, x^{(k)}}| \geq \|x^{(t_{j})}\|_p - 2 \varepsilon \|x^{(m)}\|_p\} \geq \frac{d}{2} ,
		\label{}
	\end{equation*}
	which implies the median of $|\inprod{\pi_i, x^{(k)}}|$ over $i \in [d]$ is in the range $\|x^{(t_j)}\|_p \pm 2 \varepsilon \|x^{(m)}\|_p$. In other words
    \begin{equation*}
        \|x^{(t_{j})}\|_p - 2 \varepsilon \|x^{(m)}\|_p \leq \alpha_k \leq \|x^{(t_{i})}\|_p + 2 \varepsilon \|x^{(m)}\|_p.
    \end{equation*}

	Finally we also have $\left|\|x^{(k)}\|_p - \|x^{(t_{j})}\|_p \right| \leq \varepsilon \|x^{(m)}\|_p$ by construction of the sequence $\{t_j\}_{j=1}^q$, so the claim follows up to rescaling $\varepsilon$ by a constant factor.
\end{proof}

\begin{lemma}
	\LemmaName{weak-tracking-space}
	The above algorithm can be implemented using $\Oh(\eps^{-2}\log(1/(\eps\delta))\log m)$ bits of memory to store fixed precision approximations of all counters $(\Pi x^{(k)})_i$, and $\Oh(\eps^{-p}\log(1/(\eps\delta))\log(nm))$ bits to store $\Pi$.
\end{lemma}
\begin{proof}
	Consider a sketch matrix $\Pi$ as in \Lemma{weak-tracking} --- i.e. $\Pi \in \bR^{d\times n}$ with random $\pStable$ entries, such that all rows are $r$-wise independent and all entries within a row are $s$-wise independent. Moreover let us pick some $\gamma = \Theta(\varepsilon m^{-1})$ and consider discretization $\tilde{\Pi}$ of $\Pi$, namely each entry $\tilde{\Pi}_{ij}$ is equal to $\Pi_{ij}$ rounded to the nearest integer multiple of $\gamma$. The analysis identical to the one in \cite[A.6]{KaneNW10b} shows that this discretization have no significant effect on the accuracy of the algorithm, and moreover that one can sample from a nearby distribution using only $\tau = \Oh(\log m\varepsilon^{-1})$ uniformly random bits. Therefore we can store such a matrix succinctly using $\Oh\left(r s (\log n + \tau) + r \log d\right)$ bits of memory, by storing a seed for a random $r$-wise independent hash function $h: [d] \to \{0,1\}^{\Oh(s (\log n + \tau))}$ and interpreting each $h(i)$ as a seed for an $s$-wise independent hash function describing the $i$-th row of $\tilde{\Pi}$ \cite[Corollary 3.34]{Vadhan12}. Hence the total space complexity of storing the sketch matrix $\tilde{\Pi}$ in a succinct manner is $\Oh\left( \frac{\log \delta^{-1} + \log \varepsilon^{-1}}{\varepsilon^p} ( \log n + \log m)\right)$ bits.
	
	Additionally we have to store the sketch of the current frequency vector itself, i.e. for all $i \in [d]$ we need to store $\inprod{\tilde{\pi}_i, x^{(k)}}$; for every such counter we need $\Oh(\log m \varepsilon^{-1}) = \Oh(\log m)$ bits, and there are $d = \Oh\left(\frac{\log \varepsilon^{-1} + \log \delta{-1}}{\varepsilon^{-2}}\right)$ counters.
\end{proof}

We thus have the following main theorem of this section.

\begin{theorem}
	\TheoremName{weak-tracking-alg}
	For any $p\in(0,2]$ there is an insertion-only streaming algorithm that provides the weak tracking guarantees for $f(x) = \|x\|_p$ with probability $1 - \delta$ using $\Oh\left(\frac{\log m + \log n}{\varepsilon^2} (\log \varepsilon^{-1} + \log \delta^{-1})\right)$ bits of memory.
\end{theorem}
\qed

\subsection{Strong tracking of $\|x\|_p$}\SectionName{strong-tracking}

In this section we discuss achieving a strong tracking guarantee. The same argument for $\ell_2$-tracking appeared in \cite{BravermanCINWW17}. The reduction is in fact general, and shows that for any monotone function $f$ the strong tracking problem for $f$ reduces to the weak tracking version of the same problem with smaller failure probability. 

\begin{lemma}
    \LemmaName{strong-tracking-reduction}
	Let $f : \bR^n \to \bR_+$ be any monotone functon of $\bR^n$ (i.e. $x \preceq y \implies f(x) \leq f(y)$),  such that $\min_i f(e_i) = 1$ (where $e_i$ are standard basis vectors). Let $\mathcal{A}$ be an insertion-only streaming algorithm satisfying weak tracking for any sequence of updates with probability $1 - \delta$ and accuracy $\varepsilon$. Then for a sequence of frequency vectors $0 \preceq x^{(1)} \preceq \ldots \preceq x^{(m)}$
	algorithm $\mathcal{A}$ satisfies strong tracking with probability $1 - \delta \log f(x^{(m)})$ and accuracy $2\varepsilon$.
\end{lemma}
\begin{proof}
    Define $t_1 < t_2 < \cdots < t_q$ so that $t_i$ is the smallest index in $[m]$ larger than $t_{i-1}$ with $f(x^{(t_i)}) \geq 2^i$ (if no such index exists, define $q = i$ and $t_q = m$). Note that $q \leq \log f(x^{(m)})$.

    The algorithm will fail with probability at most $\delta$ to satisfy the conclusion of \Theorem{weak-tracking-alg} for a particular sequence of vectors $x^{(1)}, x^{(2)}, \ldots x^{(t_j)}$. That is, for every $j$, with probability $1 - \delta$, we have that
    \begin{equation*}
        \forall {i \leq t_j},\ f(x^{(i)}) - \varepsilon f(x^{(t_j)}) \leq \tilde{f}^i \leq f(x^{(i)}) + \varepsilon f(x^{(t_j)}) ,
    \end{equation*}
where $\tilde{f}^t$ is the estimate output by the algorithm at time $t$.

    We can union bound over all $j \in [q]$ to deduce that except with probability $q \delta \leq \delta \log f(x^{(m)})$, 
    \begin{equation*}
        \forall {i \leq t_j},\ f(x^{(i)}) - \varepsilon f(x^{(t_j)}) \leq \tilde{f}^i \leq f(x^{(i)}) + \varepsilon f(x^{(t_j)}) .
    \end{equation*}

    By construction of the sequence of $t_j$, we know that for every $i$, if we take $t_j$ to be smallest such that $i \leq t_j$, then $f(x^{(t_j)}) \leq 2 f(x^{(i)})$, and the claim follows.
\end{proof}

\begin{theorem}
	For any $p\in(0,2]$ there is an insertion-only streaming algorithm that provides strong tracking guarantees for estimating the $\ell_p$-norm of the frequency vector with probability $1 - \delta$ and multiplicative error $1+\varepsilon$, with space usage bounded by $\Oh\left(\frac{\log m + \log n}{\varepsilon^2} (\log \varepsilon^{-1} + \log \delta^{-1} + \log \log m)\right)$ bits.
\end{theorem}
\begin{proof}
	This follows from \Lemma{weak-tracking-space} and \Lemma{strong-tracking-reduction} by observing that after a sequence of $m$ insertions, the $\ell_p$ norm of the frequency vector is bounded by $m^2$, i.e. $\log (\|x^{(m)}\|_p) = \Oh(\log m)$.
\end{proof}

\newcommand{\etalchar}[1]{$^{#1}$}

\end{document}